\documentclass[a4paper,twoside]{amsart}
\usepackage{amsmath}
\usepackage{amsfonts}
\usepackage{amssymb}
\usepackage{amsthm}
\usepackage{newlfont}
\usepackage{graphicx}
\usepackage{amscd}

\theoremstyle{plain}
\newtheorem{Th}{Theorem}[section]
\newtheorem{Cor}[Th]{Corollary}
\newtheorem{Lem}[Th]{Lemma}
\newtheorem{Prop}[Th]{Proposition}
\theoremstyle{definition}

\theoremstyle{remark}
\newtheorem*{Rem}{Remark}
\numberwithin{equation}{section}

\newcommand{\DD}{{\mathbb D}}

\newcommand{\ZZ}{{\mathbb Z}}

\newcommand{\bx}{\boldsymbol{x}}
\newcommand{\by}{\boldsymbol{y}}

\newcommand{\bphi}{\boldsymbol{\phi}}

\begin{document}

\title
{Non-commutative rational Yang--Baxter maps}

\author{Adam Doliwa}

\address{Institute of Mathematics, Polish Academy of Sciences, ul.~\'{S}niadeckich~8, 00-956 Warsaw, Poland}

\address{Faculty of Mathematics and Computer Science, University of Warmia and Mazury in Olsztyn,
ul.~S{\l}oneczna~54, 10-710~Olsztyn, Poland}
\email{doliwa@matman.uwm.edu.pl}
\urladdr{http://wmii.uwm.edu.pl/~doliwa/}

\date{}
\keywords{non-commutative integrable difference equations; functional Yang--Baxter equation; non-commutative rational maps; non-autonomous lattice Gel'fand-Dikii systems; multidimensional consistency}
\subjclass[2010]{37K10, 37K60, 16T25, 39A14, 14E07}

\begin{abstract}
Starting from multidimensional consistency of non-commutative lattice modified Gel'fand--Dikii systems we present the corresponding solutions of the functional (set-theoretic) Yang--Baxter equation, which are non-commutative versions of the maps arising from geometric crystals. Our approach works under additional condition of centrality of certain products of non-commuting variables. Then we apply such a restriction on the level of the Gel'fand--Dikii systems what allows to obtain non-autonomous (but with central non-autonomous factors) versions of the equations. In particular we recover known non-commutative version of Hirota's lattice sine-Gordon equation, and we present an integrable non-commutative and non-autonomous lattice modified Boussinesq equation. 
\end{abstract}
\maketitle

\section{Introduction}

Let $\mathcal{X}$ be any set, a map $R\colon \mathcal{X} \times \mathcal{X}$ satisfying in  
$\mathcal{X} \times \mathcal{X}  \times \mathcal{X}$ the relation
\begin{equation}
\label{eq:f-YB}
R_{12} \circ R_{13} \circ R_{23} = R_{23} \circ R_{13} \circ R_{12} , 
\end{equation}
where $R_{ij}$ acts as $R$ on the $i$-th and $j$-th factors and as identity on the third, 
is called \emph{Yang--Baxter map} \cite{Drinfeld,Veselov}. If additionally $R$ satisfies the relation
\begin{equation}
R_{21} \circ R = \mathrm{Id},
\end{equation}
where $R_{21} = \tau \circ R \circ \tau$ and $\tau$ is the transposition,
then it is called \emph{reversible Yang--Baxter map}.

In this paper we study properties of a non-commutative version of the maps arising from geometric crystals~\cite{KNY-A,Etingof,SurisVeselov}. In particular we will demonstrate the following result.
\begin{Th} \label{th:xy-YB}
Given two assemblies of (non-commuting in general) variables $\bx=(x_1, \dots, x_L)$ and $\by=(y_1, \dots , y_L)$, define polynomials
\begin{equation}
\mathcal{P}_k = \sum_{a=0}^{L-1} \left(  \prod_{i=0}^{a-1} y_{k+i} \prod_{i=a+1}^{L-1} x_{k+i} \right) , \qquad k=1,\dots , L,
\end{equation}
where subscripts in the formula are taken modulo $L$. If the products $\alpha  = x_1 x_2 \dots x_L$ and $\beta = y_1 y_2 \dots y_L$ are central then the map
\begin{equation} \label{eq:R}
R: (\bx,\by) \mapsto (\tilde{\bx}, \tilde{\by}), \qquad 
\tilde{x}_k = \mathcal{P}_{k} x_k \mathcal{P}_{k+1}^{-1}, \qquad \tilde{y}_k = \mathcal{P}_{k}^{-1} y_k \mathcal{P}_{k+1}, \qquad k=1, \dots L,
\end{equation}
is reversible Yang--Baxter map. 
\end{Th}
It is easy to see that the products $\alpha=x_1 \dots x_L$ and $\beta=y_1 \dots y_L$ are conserved by the map $R$. This can be used to reduce the number of variables. For example, in the simplest case $L=2$ define $x=x_1$, $y=y_1$ to get a parameter dependent reversible Yang--Baxter map $R(\alpha,\beta):(x,y) \mapsto (\tilde{x}, \tilde{y})$
\begin{equation}
\tilde{x} = \left(\alpha x^{-1} + y \right) x \left(x+ \beta y^{-1} \right)^{-1}, \qquad 
\tilde{y} = \left(\alpha x^{-1} + y \right)^{-1} y \left(x+ \beta y^{-1} \right),
\end{equation}
which in the commutative case is equivalent to the $F_{I\!I\!I}$ map in the list given in~\cite{ABS-YB}.

In recent studies on discrete integrable systems the property of multidimensional consistency \cite{ABS,FWN-cons} is considered as the main concept of the theory. Roughly speaking, it is the possibility of extending the number of independent variables of a given nonlinear system by adding its copies in different directions without creating this way inconsistency or multivaluedness.
It is known~\cite{ABS-YB,PTV} how to relate three dimensional consistency of integrable discrete systems with Yang--Baxter maps. There is also well known connection between Yang--Baxter maps and the braid relations. 

Non-commutative versions of integrable maps or discrete systems \cite{FWN-Capel,Kupershmidt,BobSur-nc,Nimmo-NCKP,DF-K} are of growing interest in mathematical physics. They may be considered as a useful platform to more thorough understanding of integrable quantum or statistical mechanics lattice systems, where the quantum Yang--Baxter equation~\cite{Baxter,QISM} plays a role. 

In Section~\ref{sec:nc-KP-sym} we use three dimensional consistency of non-commutative  Kadomtsev--Petviashvilii (KP) map to construct corresponding Yang--Baxter maps following ideas of \cite{KNY-A,KNY-qKP} applied there in the commutative case. It turns out that we can construct the solutions under periodicity and centrality (of certain products of the variables) assumptions. Then in Section~\ref{sec:GD-centr} we consider implication of the centrality assumption on the level of the non-commutative modified lattice Gel'fand--Dikii equations. In the simplest case we recover non-autonomous version of non-commutative Hirota's sine-Gordon equation~\cite{BobSur-nc}. We present also an integrable non-commutative and non-autonomous lattice modified Boussinesq equation.

\begin{Rem}
Throughout the paper we will work with division rings of (non-commutative) \emph{rational functions} in a finite number of (non-commuting) variables. This approach is intuitively accessible, see however \cite{Cohn} for formal definitions.
\end{Rem}

\section{Non-commutative rational realization of the symmetric group} \label{sec:nc-KP-sym}
\subsection{KP maps} \label{sec:GD} 

Consider the linear problem of the non-commutative KP hierarchy \cite{KNY-qKP,Nimmo-NCKP,Dol-GD}
\begin{equation} \label{eq:lin-KP-phi}
\bphi_{k+1} (n) - \bphi_{k} (n +\boldsymbol{\varepsilon}_i) = \bphi_{k}(n) u_{i,k}(n) ,\qquad k\in\ZZ, 
\qquad n\in\ZZ^{N}, \qquad i=1,\dots ,N,
\end{equation}
here $\bphi_{k}\colon\ZZ^N\to \DD^M$, and $\DD$ is a division ring, and $\boldsymbol{\varepsilon}_i\in\ZZ^N$ has $1$ at $i$-th place and all other zeros. 
The potentials $u_{i,k}:\ZZ^{N}\to\DD$ satisfy then the compatibility conditions
\begin{equation} \label{eq:KP-u}
u_{j,k}u_{i,k(j)}  = u_{i,k} u_{j,k(i)}  , \qquad
u_{i,k(j)}  + u_{j,k+1}  = 
u_{j,k(i)} + u_{i,k+1}, \qquad 1 \leq i\neq j \leq N,
\end{equation}
where we write $u_{i,k(j)}(n)$ instead of $u_{i,k}(n +\boldsymbol{\varepsilon}_j)$, and we skip the argument $n$.
In consequence we obtain the transformation rule
\begin{equation} \label{eq:KP-u-solved}
u_{i,k(j)} = ( u_{i,k} - u_{j,k})^{-1} u_{i,k} ( u_{i,k+1} - u_{j,k+1}), \qquad i\neq j,
\end{equation} 
which can be written as a non-commutative discrete KP map
\begin{equation*}
(\boldsymbol{u}_i, \boldsymbol{u}_j) \mapsto 
(\boldsymbol{u}_{i(j)}, \boldsymbol{u}_{j(i)} ) ,
\qquad \boldsymbol{u}_i = (u_{i,k}), \qquad k\in\ZZ. 
\end{equation*}
\begin{Prop}{\cite{Dol-GD,Dol-Des-red}}
The non-commutative discrete KP map is three-dimensionally consistent, i.e. both ways to calculate 
$\boldsymbol{u}_{i(jl)}$ give the same result, see Figure~\ref{fig:3D-GD-u}.
\end{Prop}
\begin{Rem}
Three dimensional consistency of the non-commutative discrete KP map is a consequence \cite{Dol-GD} of the four dimensional consistency of the so called Desargues maps~\cite{Dol-Des}.
\end{Rem}
\begin{figure}
\begin{center}
\includegraphics[width=12cm]{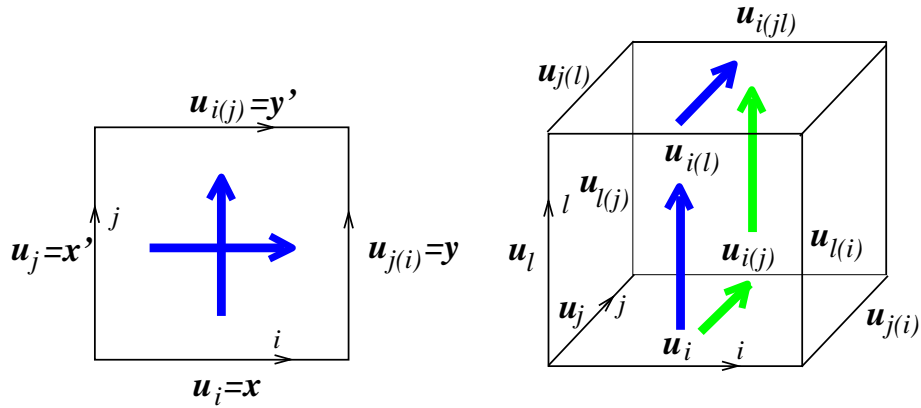}
\end{center}
\caption{The discrete KP map and its three dimensional consistency} 
\label{fig:3D-GD-u}
\end{figure}

To make connection with the Yang--Baxter maps consider $N$-cube graph, whose vertices are identified with binary sequences of length $N$ with two vertices connected by an edge if their sequences differ at one place only. The shortest paths from the initial vertex $(0,0,\dots ,0)$ to the terminal one $(1,1,\dots ,1)$ can be identified with permutations: a permutation $\sigma \in \mathcal{S}_N$ corresponds to the path with subsequent steps in directions 
$(\boldsymbol{\varepsilon}_{\sigma(1)}, \boldsymbol{\varepsilon}_{\sigma(2)}, \dots , \boldsymbol{\varepsilon}_{\sigma(N)})$. The symmetric group acts then on paths by the left natural action $\rho . \pi_\sigma = \pi_{\rho \sigma}$. 
Given initial weights $\boldsymbol{u}_i$, $i=1,\dots , N$, on edges connecting the initial vertex $(0,0,\dots ,0)$ with the vertex $\boldsymbol{\varepsilon}_i$, by the KP map we attach a weight to each edge of the cube graph. 
Each such path gives then a sequence of weights $\boldsymbol{w}^\sigma$, for example $\boldsymbol{w}^{\mathrm{Id}} = (\boldsymbol{u}_1, \boldsymbol{u}_{2(1)}, \dots , \boldsymbol{u}_{N(1,2,\dots ,N-1)})$. We are interested in maps $r_\sigma$ from the reference weights $\boldsymbol{w}^{\mathrm{Id}}$ to weights $\boldsymbol{w}^{\sigma}$. In particular, we study maps $r_i$, $i=1,\dots ,N-1$, which correspond to transpositions $\sigma_i=(i,i+1)$ generating the symmetric group $\mathcal{S}_N$ and satisfying the Coxeter relations~\cite{Humphreys}
\begin{align*}
& \sigma_j^2 = \mathrm{Id}, && \text{involutivity} \\
& \sigma_j \sigma_{j+1} \sigma_j = \sigma_{j+1} \sigma_j \sigma_{j+1}, && \text{braid relations}\\
& \sigma_i \sigma_j = \sigma_j \sigma_i \quad \text{for} \quad |i - j| > 1, && \text{commutativity}.
\end{align*}
In order to find such maps $r_i$ we have to find the so called first companion map
\begin{equation*}
(\boldsymbol{u}_i, \boldsymbol{u}_{j(i)}) \mapsto 
(\boldsymbol{u}_{j}, \boldsymbol{u}_{i(j)} ) ,
\end{equation*}
where we use variation of the terminology of~\cite{ABS-YB,PTV} where Yang--Baxter maps were studied in relation to multidimensionally consistent edge-field maps,

\subsection{The first companion map and the centrality assumption} \label{sec:comp-KP}
We will concentrate on deriving the first companion map, which we temporarily denote by $ r: (\boldsymbol{x}, \boldsymbol{y}) \mapsto (\boldsymbol{x}^\prime, \boldsymbol{y}^\prime)$, where by \eqref{eq:KP-u}
\begin{equation}
\label{eq:xjk}
x_{k}^\prime y_{k}^\prime = x_{k} y_{k}  , \qquad 
y_{k}^\prime + x_{k+1}^\prime  = y_{k} + x_{k+1}.
\end{equation}

For $\ell\in\ZZ_+$ define polynomials
\begin{equation*} \begin{split}
\mathcal{P}_{k}^{(\ell)} = \sum_{a=0}^\ell \left( \prod_{i=0}^{a-1} y_{k+i} \prod_{i=a+1}^\ell x_{k+i} \right) = & \\
x_{k+1} x_{k+2} \dots x_{k+\ell} + y_k x_{k+2} \dots x_{k+\ell-1} x_{k+\ell} + \dots + 
y_k y_{k+1} \dots & y_{k+\ell - 2} x_{k+\ell} + 
y_k \dots y_{k+\ell - 2} y_{k+\ell-1}, 
\end{split}
\end{equation*}
which satisfy the recurrence relations
\begin{equation} \label{eq:rec-P} 
\mathcal{P}_{k}^{(\ell)} = \mathcal{P}_{k}^{(\ell - 1)} \, x_{k+\ell} + \prod_{i=0}^{\ell -1} y_{k+i}  = 
\prod_{i=1}^\ell x_{k+i} + y_{k} \mathcal{P}_{k+1}^{(\ell -1)} ,
\end{equation}
where by definition $\mathcal{P}_{k}^{(0)} = 1$. By $\mathcal{P}_{k}^{(\ell)^\prime}$ denote analogous polynomials for primed variables.
\begin{Lem} \label{lem:P-P}
Assume that $x_{k}^\prime$, $y_k^\prime$ satisfy equations \eqref{eq:xjk} 
then $\mathcal{P}_{k}^{(\ell)\prime} = \mathcal{P}_{k}^{(\ell)}$.
\end{Lem}
\begin{proof} 
For $\ell = 1$ we have just the second of equations \eqref{eq:xjk}. For $\ell \geq 1$ notice that by \eqref{eq:xjk}
the product  
$\prod_{i=0}^{\ell-1} (y_{k+i} + x_{k+i+1})$ 
is equal to its primed version. 
It splits into the sum of $\mathcal{P}_{k}^{(\ell)}$ and the part with summands containing the factors 
$\dots x_{k+p} y_{k+p} \dots$ with possible $p=1,\dots ,\ell-1$. We group such unwanted terms into (disjoned) parts depending on the smallest $p$. Such a part has the structure 
\begin{equation*} 
\mathcal{P}_{k}^{(p-1)} \, x_{k+p} \, y_{k+p} \prod_{i=p+1}^{\ell-1} (y_{k+i} + x_{k+i+1}) , 
\end{equation*}
which due to the induction assumption and equations \eqref{eq:xjk} is equal to its primed version, therefore both cancel out.
\end{proof}

From now on we assume $L$-periodicity condition: $x_{k+L} = x_{k}$, $y_{k+L} = y_{k}$. 
Define $\mathcal{P}_{k} = \mathcal{P}_{k}^{(L-1)}$, then Lemma~\ref{lem:P-P} and recurrence relations 
\eqref{eq:rec-P} imply
\begin{equation*}
\mathcal{P}_{k} x_{k} + \prod_{i=0}^{L-1} y_{k+i}  = 
\prod_{i=1}^{L} x_{k+i}^\prime  + y_{k}^\prime \mathcal{P}_{k+1},\qquad
\prod_{i=1}^{L} x_{k+i} + y_{k} \mathcal{P}_{k+1} = 
\mathcal{P}_{k} x_{k}^\prime + \prod_{i=0}^{L-1} y_{k+i}^\prime  .
\end{equation*}
Notice that if we would impose the additional normalization condition 
\begin{equation} \label{eq:xy-norm}
\prod_{i=1}^{L} x_{k+i}^\prime    = \prod_{i=0}^{L-1} y_{k+i}  , \qquad 
\prod_{i=0}^{L-1} y_{k+i}^\prime = \prod_{i=1}^{L} x_{k+i},
\end{equation}
then equations \eqref{eq:xjk} could be solved as
\begin{equation} \label{eq:xy-j-sol}
x_{k}^\prime = \mathcal{P}_{k}^{-1} y_{k} \mathcal{P}_{k+1}, \qquad 
y_{k}^\prime = \mathcal{P}_{k} x_{k} \mathcal{P}_{k+1}^{-1}. 
\end{equation}

However, equations \eqref{eq:xy-j-sol} and condition \eqref{eq:xy-norm} are not compatible for general non-commuting variables. The above procedure of getting solutions works if we make additional centrality assumptions which state that 
$\alpha = \prod_{i=1}^{L} x_{i}$ and $\beta = \prod_{i=1}^{L} y_{i}$ commute with other elements of the division ring. 
\begin{Lem} \label{lem:XY-central}
Under the centrality assumptions the products $\prod_{i=1}^{L} x_{k+i}$ and $ \prod_{i=1}^{L} y_{k+i}$ do not depend on the index $k$. Moreover
\begin{equation} 
\mathcal{P}_{k} x_{k} - y_{k} \mathcal{P}_{k+1} = \alpha -  \beta,	
\end{equation}
which means that the above expression is central and independent of index $k$ as well. In particular $\mathcal{P}_{k} x_{k} $ commutes with $y_{k} \mathcal{P}_{k+1}$.
\end{Lem}
\begin{proof}
The first part follows from identities
\begin{equation*}
\prod_{i=1}^{L} x_{k+i} = (x_1 \dots x_{k-1})^{-1} \, \alpha \, (x_1 \dots x_{k-1}) , \qquad 
\prod_{i=1}^{L} y_{k+i} = (y_1 \dots y_{k-1})^{-1} \, \beta \, (y_1 \dots y_{k-1}) ,
\end{equation*}
where we used also the periodicity assumption. The second part is implied by 
equations \eqref{eq:rec-P}.
\end{proof}
\begin{Prop}
Under the centrality assumption the expressions for
$x_{k}^\prime $ and $y_{k}^\prime$ given by \eqref{eq:xy-j-sol} provide the unique solution of equations \eqref{eq:xjk} supplemented by the normalization conditions $\alpha^\prime  = \beta$ and $\beta^\prime  = \alpha$.
\end{Prop}
\begin{proof} Notice that by Lemma~\ref{lem:XY-central} $x_{k}^\prime $ and $y_{k}^\prime$ given by \eqref{eq:xy-j-sol} satisfy the normalization condition.
Then also both expressions 
\begin{equation*}
y_{k}^\prime + x_{k+1}^\prime - y_{k} - x_{k+1} =
\left( \mathcal{P}_{k} x_{k} - y_{k} \mathcal{P}_{k+1} \right)\mathcal{P}_{k+1}^{-1} +
\mathcal{P}_{k+1}^{-1} \left( y_{k+1} \mathcal{P}_{k+2}  - \mathcal{P}_{k+1} x_{k+1} \right)
\end{equation*}
and
\begin{equation*}
x_{k}^\prime y_{k}^\prime  - x_{k} y_{k} =   \mathcal{P}_{k}^{-1} \left( y_{k} \mathcal{P}_{k+1} \, \mathcal{P}_{k} x_{k} -
\mathcal{P}_{k} x_{k} \, y_{k}  \mathcal{P}_{k+1} \right) \mathcal{P}_{k+1}^{-1} 
\end{equation*}
vanish due to Lemma~\ref{lem:XY-central}.
\end{proof}       
\begin{Cor} \label{cor:inv}
The first companion map $ (\boldsymbol{x}, \boldsymbol{y}) \mapsto (\boldsymbol{x}^\prime, \boldsymbol{y}^\prime)$ given above is involutory.
\end{Cor}
\begin{Cor} \label{cor:braid}
The problem of finding the first companion of the KP map in the periodic reduction can be considered as a refactorization problem 
$A(\bx) A(\by) = A(\bx^\prime) A(\by^\prime)$, where the matrix
\begin{equation}
A(\bx) = \left(  \begin{array}{ccccc}  
- x_{1} & 0  & \cdots & 0 & \lambda \\
1 & - x_{2}  & 0 & \hdots  & 0 \\
0 & 1 & \ddots &  &  \vdots \\
\vdots &  & & -x_{L-1}  & 0 \\
0 & 0 & \ \hdots  & 1 & - x_{L}  \end{array} \right) 
\end{equation}
with the central spectral parameter $\lambda$ is the  $L$-periodic reduction of the discrete non-commutative KP hierarchy (Gel'fand--Dikii system) linear problem \eqref{eq:lin-KP-phi} studied in \cite{Dol-GD}.
\end{Cor}

\subsection{Realization of Coxeter relation under the centrality assumption}
Consider again a sequence $(\boldsymbol{w}_1, \boldsymbol{w}_2, \dots , \boldsymbol{w}_N) $ o weights along a shortest path in $N$-cube from the initial to the terminal vertex, each weight is a sequence $\boldsymbol{w}_j= (w_{j,1}, \dots w_{j,L} )$ of non-commuting variables satisfying the centrality assumption that the product $\alpha_j = w_{j,1} w_{j,2}  \dots w_{j,L}$ commutes with all $w_{j,k}$. As we already have mentioned the symmetric group $\mathcal{S}_N$ acts in natural way on the paths and thus on the weights. To make use of results of Section~\ref{sec:comp-KP} define polynomials
\begin{equation} 
\mathcal{P}_{j,k} = \sum_{a=0}^{L-1} \left( \prod_{i=0}^{a-1} w_{j+1,k+i} \prod_{i=a+1}^{L-1} w_{j,k+i} \right) ,
\end{equation}
where the second index should be considered modulo $L$. 
\begin{Prop}
Define the rational maps $r_j$, $j=1,\dots , N-1$, of the non-commuting variables $w_{j,k}$, $j=1,\dots , N$, $k=1,\dots ,L$, 
\begin{align} \label{eq:r-j-xij}
r_j(w_{j,k}) \, & = P_{j,k}^{-1} \, w_{j+1,k} \, P_{j,k+1}, \\ 
r_j(w_{j+1,k}) & = P_{j,k} \, w_{j,k} \, P_{j,k+1}^{-1}, \\ 
r_j(w_{i,k}) \, &  = w_{i,k} \quad \qquad  \text{for} \quad i\neq j, j+1.
\end{align}
If we assume centrality of the products $\alpha_j = \prod_{i=1}^L w_{j,i}$ then the maps $r_j$ satisfy the Coxeter relations
\begin{equation*}
r_j^2 = \mathrm{Id}, \qquad
r_j \circ r_{j+1} \circ r_j = r_{j+1} \circ r_j \circ r_{j+1},\qquad
r_i \circ r_j = r_j \circ r_i \qquad  \text{for} \quad |i - j| > 1.
\end{equation*}
\end{Prop}
\begin{proof}
The commutativity part is clear from definition of the maps. Involutivity part comes from Corollary~\ref{cor:inv}, and the 
braid relations follow from the path-interpretation  and uniqueness of the companion map subject to normalization conditions~\eqref{eq:xy-norm}. Equivalently, we can use the unique refactorization interpretation given in Corollary~\ref{cor:braid}, and follow the argumentation presented in \cite{Veselov}.
\end{proof}
\begin{Cor}
The action of $r_j$ on the central elements $\alpha_i$ is 
\begin{equation}
r_j(\alpha_j) = \alpha_{j+1}, \qquad r_j(\alpha_{j+1}) = \alpha_j, \qquad r_j(\alpha_i) = \alpha_i
\quad \text{for} \quad i\neq j, j+1.
\end{equation}
\end{Cor}
\begin{Cor}
We can consider the division ring $\DD$ as a division algebra over a fixed subfield $\Bbbk$ of its center. Therefore we can state the centrality condition as $\alpha_j\in\Bbbk$.
\end{Cor}
\begin{Cor}
Define second companion map $R = \tau \circ r$, i.e. 
$ R: (\boldsymbol{x}, \boldsymbol{y}) \mapsto (\tilde{\boldsymbol{x}}, \tilde{\boldsymbol{y}})= (\boldsymbol{y}^\prime, \boldsymbol{x}^\prime)$ then by \eqref{eq:xy-j-sol}
we obtain formulas \eqref{eq:R}. Due to well known relation~\cite{Veselov} between realizations of the Coxeter relations and reversible Yang--Baxter maps we proved in this way Theorem~\ref{th:xy-YB}.
\end{Cor}
\begin{figure}
\begin{center}
\includegraphics[width=14cm]{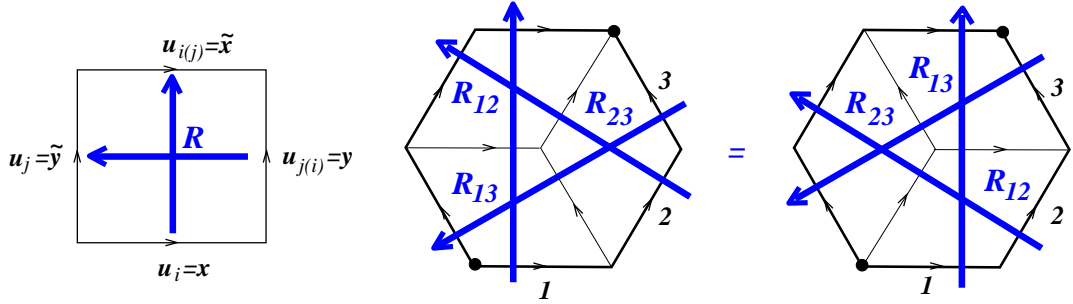}
\end{center}
\caption{Second companion map as a Yang--Baxter map} 
\label{fig:3D-GD-YB-2}
\end{figure}

\section{Non-commutative Gel'fand--Dikii systems with the centrality condition} \label{sec:GD-centr}
In \cite{Dol-GD,Dol-Des-red} we studied periodic reductions of the discrete KP hierarchy under two extreme assumptions about non-commutativity/commutativity of dependent variables. Results of Section~\ref{sec:comp-KP} suggest to consider an analogous centrality condition on the level of equations~\eqref{eq:KP-u}. By simple calculation we obtain the following result.
\begin{Prop}
In the $L$-periodic reduction $u_{i,k+L} = u_{i,k}$ of the non-commutative KP system~\eqref{eq:KP-u} assume centrality of the products $\mathcal{U}_i = u_{i,1} u_{i,2} \dots u_{i,L}$. Then the products $u_{i,k} u_{i,k+1} \dots u_{i,k+L-1}$ do not depend on index $k$, and $\mathcal{U}_i$ is a function of $n_i$ only
\begin{equation}
\mathcal{U}_{i(j)} = \mathcal{U}_i, \qquad j\neq i.
\end{equation}
\end{Prop}
Using the above result one can obtain non-autonomous non-commutative discrete equations of the modified Gel'fand--Dikii type. It is known~\cite{Dol-GD} that the first part of equations\eqref{eq:KP-u} implies existence of potentials $\rho_k$ such that $u_{i,k} = \rho_k^{-1} \rho_{k(i)}$, while the second part gives the corresponding vertex form of the non-commutative discrete KP hierarchy
\begin{equation} \label{eq:KP-r}
(\rho_{k(j)}^{-1} - \rho_{k(i)}^{-1}) \rho_{k(ij)} = \rho_{k+1}^{-1} (\rho_{k+1(i)} - \rho_{k+1(j)} ), \qquad k\in\ZZ / (L\ZZ), \qquad i\neq j.
\end{equation}
Then we replace one of the functions $\rho_i$ by others and the central non-autonomous factors. To make connection with known results it is convenient to define central functions $\mathcal{F}_i = \left( \mathcal{U}_i \right)^{1/L}$ of the corresponding single variables $n_i$, and then consider the central function $\mathcal{G}$ defined by compatible system $\mathcal{G}_{(i)} = \mathcal{F}_i \mathcal{G}$. We remark that such $\mathcal{G}$ is a product of functions of single variables. 

In the simplest case $L=2$ define, like in \cite{Dol-GD}, a function $x$ by $\rho_1 = x \mathcal{G}$. Then 
\begin{equation*}
u_{i,1} = \rho_1^{-1} \rho_{1(i)} = x^{-1} x_{(i)} \mathcal{F}_i, 
 \qquad \text{and} \qquad
u_{i,2} = \rho_2^{-1} \rho_{2(i)} = x_{(i)}^{-1} x \mathcal{F}_i,
\end{equation*}
which inserted in equations~\eqref{eq:KP-r} produces the non-commutative Hirota (or discrete sine-Gordon or lattice modified Korteweg--de~Vries) equation studied in~\cite{Hirota-KdV,Hirota-sG,BobSur-nc}
\begin{equation} \label{eq:Hirota-SG}
\left( x_{(j)}^{-1} \mathcal{F}_i - x_{(i)}^{-1} \mathcal{F}_j \right) x_{(ij)} = 
\left( x_{(i)}^{-1} \mathcal{F}_i - x_{(j)}^{-1} \mathcal{F}_j \right) x .
\end{equation}
\begin{Rem}
To recover the equation in the form studied in~\cite{BobSur-nc} notice that after extracting $x_{(j)}^{-1}$ the expressions in brackets commute, and use inverses of the non-autonomous factors  $\mathcal{F}_i$.
\end{Rem}

For $L=3$ define unknown functions $x$ and $y$ by equations
\begin{equation*}
\rho_1 = x \mathcal{G} , \qquad \rho_3 = y^{-1}\mathcal{G}, \qquad \mathcal{G}_{(i)} = \mathcal{F}_i \mathcal{G}
\end{equation*}
which allows to find
\begin{equation*}
\rho_2^{-1} \rho_{2(i)} = u_{i,2} = x_{(i)}^{-1} x y_{(i)} y^{-1} \mathcal{F}_i .
\end{equation*}
Making such substitution in \eqref{eq:KP-r} for $k=1$ and $k=3$ we obtain the following non-commutative integrable two-component system (equation for $k=2$ is then its consequence)
\begin{align*}
\bigl( x_{(j)}^{-1} \mathcal{F}_{i} - x_{(i)}^{-1} \mathcal{F}_{j} \bigr) x_{(ij)} & =
\bigl(  x_{(i)}^{-1} x y_{(i)} \mathcal{F}_{i} -  x_{(j)}^{-1} x y_{(j)} \mathcal{F}_{j} \bigr) y^{-1}, \\
\bigl( y_{(j)} \mathcal{F}_{i} - y_{(i)} \mathcal{F}_{j} \bigr) y^{-1}_{(ij)} & =
x^{-1} \bigl(  x_{(i)} \mathcal{F}_{i} -  x_{(j)} \mathcal{F}_{j} \bigr) .
\end{align*}
Next, by elimination of the field $x$ we obtain integrable non-commutative and non-autonomous (with central non-autonomous coefficients $\mathcal{F}_{i}$) version of the lattice modified Boussinesq~\cite{FWN-lB} equation
\begin{equation*} \begin{split}
\left[ \left( y_{(j)} \mathcal{F}_{i} - y_{(i)} \mathcal{F}_{j} \right)^{-1} y^{-1}_{(ij)} \right]_{(ij)} - &
y \bigl( y_{(i)}^{-1} \mathcal{F}_{i} - y_{(j)}^{-1} \mathcal{F}_{j} \bigr) = \\ 
\left[  y_{(ij)} \left( y_{(j)} \mathcal{F}_{i} -  y_{(i)} \mathcal{F}_{j} \right)^{-1} 
\left( y_{(i)} \mathcal{F}_{i}^2 -   y_{(j)} \mathcal{F}_{j}^2 \right) y^{-1} \right]_{(i)}  - &
\left[  y_{(ij)} \left( y_{(j)} \mathcal{F}_{i} -  y_{(i)} \mathcal{F}_{j} \right)^{-1} 
\left( y_{(i)} \mathcal{F}_{i}^2 -   y_{(j)} \mathcal{F}_{j}^2 \right) y^{-1} \right]_{(j)} .
\end{split} 
\end{equation*}

\section{Concluding remarks}
We presented a non-commutative rational Yang--Baxter map obtained from the non-commutative discrete KP hierarchy subject to periodicity and centrality constraints. The corresponding integrable systems, which generalize the non-commutative non-autonomous Hirota's sine-Gordon equation~\cite{BobSur-nc} have been also considered. In particular we have obtained an integrable non-commutative and non-autonomous lattice modified Boussinesq equation. We remark, see~\cite{Dol-GD,Dol-Des-red}, that three dimensional consistency of the equations considered here is a consequence of four dimensional compatibility of the non-commutative Hirota's discrete KP system \cite{Dol-Des}, where the counterpart of the functional Yang--Baxter equation is the functional pentagon equation~\cite{DoliwaSergeev-pentagon}. Since the solutions of the pentagon equation presented in~\cite{DoliwaSergeev-pentagon} allow for quantization (understood as a reduction from the non-commutative case by adding certain commutation relations preserved by the integrable evolution), we expect that also the non-commutative rational Yang--Baxter map obtained above can be quantized in such a way also. It would be instructive to understand various applications of the Hirota discrete KP systems and its reductions reviewed in~\cite{KNS-rev} from that perspective. 

\section*{Acknowledgments}
The research was initiated during author's work at Institute of Mathematics of the Polish Academy of Sciences.
The paper was supported in part by Polish Ministry of Science and Higher Education grant No.~N~N202~174739.

\bibliographystyle{amsplain}

\providecommand{\bysame}{\leavevmode\hbox to3em{\hrulefill}\thinspace}

\end{document}